\documentclass[11pt]{article}

\usepackage[margin=1in]{geometry}
\usepackage{amsmath,amssymb,amsthm,mathtools}
\usepackage{thmtools,thm-restate}
\usepackage{framed}
\newcommand{\yijunresolved}[1]{}
\newcommand{\claude}[1]{}
\usepackage{enumitem}
\usepackage{microtype}
\usepackage{url}
\usepackage{xcolor}
\usepackage[hidelinks]{hyperref}
\usepackage[nameinlink,noabbrev]{cleveref}
\crefname{theorem}{Theorem}{Theorems}
\crefname{lemma}{Lemma}{Lemmas}
\crefname{observation}{Observation}{Observations}
\crefname{proposition}{Proposition}{Propositions}
\crefname{corollary}{Corollary}{Corollaries}
\crefname{section}{Section}{Sections}
\crefname{equation}{Equation}{Equations}

\newtheorem{theorem}{Theorem}
\newtheorem{lemma}[theorem]{Lemma}
\newtheorem{proposition}[theorem]{Proposition}
\newtheorem{corollary}[theorem]{Corollary}

\newtheorem{observation}[theorem]{Observation}
\theoremstyle{definition}
\newtheorem{definition}[theorem]{Definition}
\theoremstyle{remark}

\newcommand{\eps}{\varepsilon}
\newcommand{\EE}{\mathbb{E}}
\newcommand{\Prb}{\mathbb{P}}

\newcommand{\vbl}{\operatorname{vbl}}
\newcommand{\poly}{\operatorname{poly}}
\newcommand{\1}{\mathbf{1}}

\newif\ifshowcomments
\showcommentstrue

\newcommand{\makecomment}[3]{%
  \ifshowcomments
    \expandafter\newcommand\csname #1\endcsname[1]{%
      \textcolor{#3}{\textsf{\textbf{[#2:} ##1\textbf{]}}}}%
  \else
    \expandafter\newcommand\csname #1\endcsname[1]{}%
  \fi
}

\makecomment{hung}{Hung}{red!80!black}
\makecomment{nima}{Nima}{blue!80!black}
\makecomment{yijun}{Yijun}{green!45!black}
\makecomment{todonote}{TODO}{orange!85!black}

\title{Communication-Efficient $(1+\eps)\Delta$-Edge Coloring \\
and Lov\'asz Local Lemma}

\author{Yi-Jun Chang\footnote{National University of Singapore. ORCID: 0000-0002-0109-2432. Email: cyijun@nus.edu.sg}  \and Nima Dolatabadi\footnote{University of Copenhagen. Supported by VILLUM Foundation grant 54451, Basic Algorithms Research Copenhagen (BARC). Most of this work was completed while the author was at the National University of Singapore. ORCID: 0009-0000-0928-7499. Email: nima.dolatabadi@di.ku.dk}  \and Hung Thuan Nguyen\footnote{National University of Singapore. ORCID: 0009-0006-7993-2952.  Email: hung@u.nus.edu}}

\date{}

\begin{document}
\renewcommand{\yijun}[1]{{\color{red!75!black}\small\textsf{[Yi-Jun, OPEN] #1}}}
\maketitle

\begin{abstract}
We study edge coloring in the two-party edge-partition model, where Alice and
Bob each know part of the edge set and must jointly produce a proper coloring
with little communication. Previous work gave a deterministic
$(2\Delta-1)$-edge-coloring protocol using $O(n)$ bits, leaving open whether fewer
colors can be obtained efficiently.

We simultaneously reduce both the number of colors and the communication. For every fixed $\eps>0$ and all sufficiently large $\Delta$, we give a
public-coin Las Vegas protocol that finds a proper
$(1+\eps)\Delta$-edge coloring using
$O(ne^{-\gamma\Delta} + 1)$ expected bits and
$O\left(\frac{\log n}{\Delta}+1\right)$ expected rounds, where $\gamma>0$
depends only on $\eps$. 
Thus, the expected communication is $o(n)$ when $\Delta=\omega(1)$
and $O(1)$ when $\Delta\ge C_\varepsilon\log n$,
for a sufficiently large constant $C_\varepsilon$.
 Using only
private coins adds $O(\log n)$ expected bits.

Our key idea is a new randomized coloring procedure that allows Alice and Bob
to color their edges using essentially the same color space with only a small
amount of coordination, so most of their random choices remain private. To
make this procedure succeed, we develop a communication-efficient constructive
Lov\'asz local lemma (LLL) for two parties.  

Our two-party constructive LLL is also of independent interest. We illustrate
its broader applicability by applying it to standard LLL formulations of
several other classical problems, obtaining communication-efficient two-party
protocols.
\end{abstract}

\thispagestyle{empty}
\newpage
\bigskip
\tableofcontents
\bigskip
\thispagestyle{empty}

\newpage
\pagenumbering{arabic}

\section{Introduction}
\label{sec:introduction}

The two-party communication model~\cite{Yao79,KN97} asks how much
information two parties need to exchange in order to solve a problem when
each party sees only part of the input. Local computation is free, and the
main measure of complexity is the total number of bits communicated; we also
consider the number of communication rounds. We count each message sent by
one party to the other as one round, with Alice and Bob alternating as the
sender.

Randomized protocols may use either \emph{public randomness}, visible to both parties, or only \emph{private randomness}, where each party has its own independent random bits. A \emph{Las Vegas} protocol always outputs a correct solution and terminates with probability one; its communication and number of rounds are measured in expectation. By Newman's theorem~\cite{Newman91,KN97}, any public-coin Las Vegas protocol on a finite input space can be converted to a private-coin protocol with only an additional $O(\log\log K)$ expected bits of communication, where $K$ is the number of possible input pairs.

For graph problems, a natural setting is the \emph{edge-partition model}. Both parties know a common vertex set $V$, while the edge set $E$ of a graph $G=(V,E)$ is partitioned into two disjoint sets $E_A$ and $E_B$. Alice knows $E_A$, Bob knows $E_B$, and neither party knows the edges held by the other. The two parties communicate in order to jointly solve a problem on $G$. This model and closely related two-party graph models have been studied for connectivity, matching, shortest paths, cut problems, and graph coloring~\cite{HajnalMaassTuran88,IvanyosKlauckLeeSanthaWolf12,BlikstadBrandEfronMukhopadhyayNanongkai22,BlikstadJiangMukhopadhyayYingchareonthawornchai25,ChengJSY26,JiangNanongkaiSawettamalya26,RadhakrishnanReddyVenkat26,FlinMittal25,ChangMNS25}.

\subsection{Edge coloring}
Our main problem is edge coloring. A \emph{proper edge coloring} assigns a color to every edge so that any two edges sharing an endpoint receive different colors. In the edge-partition model, Alice colors the edges in $E_A$ and Bob colors the edges in $E_B$, and the union of their outputs must form a proper edge coloring of the entire graph $G=(V, E_A \cup E_B)$. Throughout the paper, we assume that the maximum degree $\Delta$ of $G$ is known to both parties.

For a simple graph of maximum degree $\Delta$, at least $\Delta$ colors may be necessary, while Vizing's theorem~\cite{Vizing64} guarantees that $\Delta+1$ colors always suffice. We ask how closely Alice and Bob can approach this bound without having to exchange a large amount of information about their respective edge sets.

Chang, Mishra, Nguyen, and Salim~\cite{ChangMNS25} initiated the study of edge coloring in the two-party edge-partition model, where Alice and Bob each color their own edges and the union of their outputs must be a proper edge coloring of the entire graph. They showed that $2\Delta$ colors can be achieved without any communication, while $2\Delta-1$ colors can be achieved deterministically using $O(n)$ bits and $O(1)$ rounds. Moreover, for constant $\Delta$, they proved a matching $\Omega(n)$ lower bound for $(2\Delta-1)$-edge coloring. Thus, $2\Delta-1$ is exactly the largest palette size for which linear communication can be necessary: with one additional color, their zero-communication protocol already applies.

This leaves a large gap between the $2\Delta-1$ colors achievable with linear communication and the $\Delta+1$ colors guaranteed by Vizing's theorem. They identified reducing the number of colors as a natural open direction, asking in particular for the communication complexity of $(\Delta+1)$-edge coloring.

In this work, we show that one can approach the Vizing bound while maintaining the $O(n)$ communication complexity. Perhaps surprisingly, substantially less communication is possible when the maximum degree grows: our communication complexity becomes $o(n)$ whenever $\Delta=\omega(1)$.

\begin{restatable}[Two-party edge coloring]{theorem}{thmmain}
\label{thm:main}
For every fixed $\eps\in(0,1)$, there is a constant
$\Delta_0=\Delta_0(\eps)$ such that, for every $n$-vertex simple graph
$G$ of maximum degree at most $\Delta\ge\Delta_0$, there is a Las Vegas
two-party protocol that produces a proper
$(1+\eps)\Delta$-edge coloring in
$O\left(\frac{\log n}{\Delta}+1\right)$ expected rounds.
For some constant $\gamma>0$ depending only on $\eps$, the expected
communication is $O(ne^{-\gamma\Delta} + 1)$ with public randomness and
$O(ne^{-\gamma\Delta}+\log n)$ using only private randomness.
\end{restatable}

The communication bound has a particularly sharp dependence on the maximum
degree: once
$\Delta=\Omega(\log n)$ is sufficiently large, $O(1)$ bits of communication suffice! Thus the $\Omega(n)$ lower bound of Chang, Mishra,
Nguyen, and Salim~\cite{ChangMNS25} for constant $\Delta$ cannot extend uniformly to larger degrees.
\Cref{tab:results} compares our result with the upper and lower bounds of
Chang, Mishra, Nguyen, and Salim~\cite{ChangMNS25}.

\begin{table}[ht!]
\centering
\begin{tabular}{lllll}
\hline
& Colors & Communication & Rounds & Randomness\\
\hline
\cite{ChangMNS25}
& $2\Delta$
& $0$
& $0$
& Deterministic\\
\cite{ChangMNS25}
& $2\Delta-1$
& $O(n)$
& $O(1)$
& Deterministic\\
\cite{ChangMNS25}
& $2\Delta-1$
& $\Omega(n)$ for constant $\Delta$
& --
& Randomized\\
\Cref{thm:main}
& $(1+\eps)\Delta$
& $O(ne^{-\gamma\Delta} + 1)$
& $O\left(\frac{\log n}{\Delta}+1\right)$
& Public-coin Las Vegas\\
\Cref{thm:main}
& $(1+\eps)\Delta$
& $O(ne^{-\gamma\Delta}+\log n)$
& $O\left(\frac{\log n}{\Delta}+1\right)$
& Private-coin Las Vegas\\
\hline
\end{tabular}
\caption{Edge coloring in the two-party edge-partition model.}
\label{tab:results}
\end{table}

\subsection{Constructive Lov\'asz local lemma}
\label{sec:results-lll}

Our main technical contribution is a communication-efficient constructive
Lov\'asz local lemma (LLL) for two parties. Roughly speaking, Alice and Bob
each hold a collection of bad events, and their goal is to find an assignment
avoiding all of them while exchanging little information.

We use the standard variable model of the LLL. The instance contains a finite
collection $\mathcal V$ of mutually independent random variables, partitioned
as
$$
\mathcal V=\mathcal X\cup\mathcal V_A\cup\mathcal V_B.
$$
The variables in $\mathcal X$ are shared. Their indexing, domains, and
distributions are known to both parties, and Alice and Bob must agree on
their final values. The variables in $\mathcal V_A$ and $\mathcal V_B$ are
private to Alice and Bob, respectively. Each party knows the domains and
distributions of its own private variables and can sample from them.

The family of bad events is partitioned as
$$
\mathcal B=\mathcal B_A\cup\mathcal B_B.
$$
Alice knows the description of every event in $\mathcal B_A$, including the
variables on which it depends, and can determine whether it occurs from the
current values of these variables. Every such event depends only on variables
in $\mathcal X\cup\mathcal V_A$. Symmetrically, Bob knows and can evaluate
the events in $\mathcal B_B$, each of which depends only on variables in
$\mathcal X\cup\mathcal V_B$. Neither party needs to know the other party's
bad events or private variables.

The goal is to assign values to all variables so that no bad event occurs.
At the end of the protocol, the parties agree on the values of the shared
variables, while each party needs to know only the values of its own private
variables.

For an event $B$, let $\vbl(B)$ denote the set of variables on which it
depends. The \emph{dependency graph} has vertex set $\mathcal B$, with an
edge between distinct bad events $B,B'\in\mathcal B$ if
$\vbl(B)\cap\vbl(B')\neq\varnothing$. The dependency graph need not be known to either party. We write $\Gamma(B)$ for the
neighborhood of $B$ in this graph. We use the following five parameters, which are known to both parties:
\begin{align*}
N &= |\mathcal X|,
&& \text{the number of shared variables},\\
M &\ge |\mathcal B|,
&& \text{an upper bound on the number of bad events},\\
p &\ge \max_{B\in\mathcal B}\Pr[B],
&& \text{an upper bound on the probability of any bad event},\\
D &\ge \max_{B\in\mathcal B}|\Gamma(B)|,
&& \text{an upper bound on the degree of the dependency graph},\\
h &\ge \max_{B\in\mathcal B}|\vbl(B)\cap\mathcal X|,
&& \text{an upper bound on the number of shared variables in any bad event}.
\end{align*}

Without the partition between Alice and Bob, this is the usual variable
setting of the Lov\'asz local lemma. In particular, the symmetric LLL
guarantees that an assignment avoiding all bad events exists when
$ep(D+1)\le 1$.

Our two-party LLL requires only constant slack $\delta$ in this condition, and its
communication bound depends on the number of shared variables but not on the
sizes of their domains.

For the theorem below, we may assume that $Mph>0$. Indeed, if $M=0$, there
are no bad events, while if $p=0$, the initial random assignment avoids all
bad events with probability one. If $h=0$, then no bad event depends on a
shared variable. Under the LLL condition below, Alice and Bob can run the
Moser--Tardos algorithm independently on their respective bad events, and
their two assignments together avoid all bad events. In all three cases, no
communication is needed. Throughout the remainder of the paper, we therefore
assume that $Mph>0$.

\begin{restatable}[Two-party constructive LLL]{theorem}{thmtwoparty}
\label{thm:two-party-lll}
Fix any constant $\delta\in(0,1)$, and suppose
$$
ep(D+1)\le 1-\delta.
$$
There is a public-coin Las Vegas protocol that finds an assignment
avoiding all bad events. The expected communication complexity is
$$
O\left(
1+Mph\log\left(2+\frac{N}{Mph}\right)
\right).
$$
The expected number of rounds is
$$
O\left(
1+\frac{\log(1+A)}{\log(1/\lambda)}
\right),
\quad
\text{where} \quad\lambda=ep(D+1) \quad
\text{and} \quad
A=\frac{Mp}{1-\lambda}.
$$
\end{restatable}

We state the theorem only for public randomness. By Newman's theorem, it can
be converted to a private-coin protocol with an additional
$O(\log\log K)$ expected bits of communication, where $K$ is the number of
possible input pairs. Unlike the parameters $(N,M,p,D,h)$ above, however,
$K$ depends on how the LLL instance is represented in a particular
application. We therefore account for this overhead separately when applying
the theorem.

A particularly useful regime is $Mph=O(N)$. In this case, the expected
communication complexity is simply $O(N)$: up to a constant factor, the
parties communicate only one bit per shared variable, even if the shared
variables have arbitrarily large domains.

Our two-party LLL protocol is useful well beyond edge coloring. We apply \Cref{thm:two-party-lll} to several
other classical problems with natural LLL formulations, obtaining
communication-efficient protocols.

\subsection{Technical overview}
\label{sec:overview}

We explain the main ideas behind our two-party communication protocols for LLL and edge coloring.

\paragraph{Communication-efficient Moser--Tardos resampling.}

Our LLL protocol builds on the Moser--Tardos resampling
algorithm~\cite{MoserTardos10}, which repeatedly chooses a bad event that
currently occurs and resamples all variables on which it depends. Directly
simulating this algorithm in the two-party setting would require too much
communication. After Alice resamples one of her events, Bob may need to learn
the new shared values before evaluating his own events, so the parties would
have to synchronize after essentially every resampling. Moreover, transmitting
a shared value can itself be expensive when its domain is large; in our
edge-coloring application, some shared variables are permutations.

Our key departure from this direct simulation is to make the resampling \emph{exhaustive} between communications. Alice repeatedly resamples her occurring bad events until none remains, and only then sends Bob the net change in the shared state. Bob does the same, and the parties alternate. Intermediate states during Alice's round are irrelevant to Bob, since he evaluates his events only after Alice has finished. At the level of individual resamplings, this is still an ordinary Moser--Tardos execution, using a particular adaptive rule for choosing the next violated event.

To avoid transmitting values from large domains, we expose the shared randomness in advance. For every shared variable $X_i$, public randomness generates an infinite resampling table $X_i(0),X_i(1),X_i(2),\ldots$, and a counter $c_i$ specifies its current entry. The parties therefore synchronize only the changes in the counters; the corresponding values can be reconstructed from public randomness. This makes the communication independent of the domain sizes of the shared variables.

The main question is why this exhaustive repair process does not require Alice
and Bob to repeatedly undo each other's progress. Suppose Alice finishes her
round with none of her bad events occurring, but some Alice event occurs again
after Bob's next round. Bob cannot change Alice's private variables, so one of
his resamplings must have changed a shared variable on which this event
depends. Similarly, that Bob resampling can be traced back to an earlier
resampling by Alice. Continuing backwards produces a long alternating chain
of dependent resamplings. Thus, any resampling that occurs in a late round
must be supported by a correspondingly long chain of dependencies.

The standard Moser--Tardos witness-tree analysis captures precisely such
causal chains. Under the LLL condition with slack $ep(D+1)\le 1-\delta$, deep witness trees become
geometrically less likely as their depth increases. Consequently, the
expected number of resamplings decreases geometrically from round to round.
Intuitively, every additional alternation between Alice and Bob requires one
more link in the causal chain, making continued interaction increasingly
unlikely.

It remains to communicate the changes to the shared counters efficiently. In
a round, suppose the shared counters increase a total of $U$ times. Their net
changes form a nonnegative vector with $N$ coordinates and coordinate sum
$U$. If $U=0$, the sender communicates only that the update is empty. If
$U>0$, there are
$\binom{N+U-1}{U}$ possible update vectors. After first communicating $U$,
the sender encodes the update vector by its rank among these possible vectors.
Including the encoding of $U$, this uses
$O\left(U\log\left(e\left(1+N/U\right)\right)\right)$ bits.

Since each resampling increases at most $h$ shared counters, and the expected number of
resamplings decreases geometrically across rounds, summing these costs gives
the communication bound in \Cref{thm:two-party-lll}.

\paragraph{Edge coloring with a small shared interface.}
We next design a randomized edge-coloring procedure to which we can apply our
two-party LLL. Besides making the bad events rare and weakly dependent, we
must ensure that the procedure uses only a small number of shared random
variables.

This requirement is specific to the communication setting. In distributed computing, existing randomized approaches~\cite{DubhashiGrablePanconesi98,ElkinPettieSu15,ChangHLPU18} to $(1+\eps)\Delta$-edge coloring, based on the nibble method and LLL, make random choices associated with individual edges. A direct adaptation creates a dilemma. If the random choices of Alice's and Bob's edges remain private, then a conflict between an Alice edge and a Bob edge depends on private randomness from both parties and does not fit our partitioned LLL model. If instead all edge-level choices are made shared, the resulting LLL may have one shared variable per edge, which is too expensive for the communication bounds we seek.

Our construction resolves this dilemma by separating the randomness used to
coordinate the two parties from the randomness used to balance their local
loads. Only the coordination randomness is shared. We partition most of the
color space into a constant number $q=q(\eps)$ of regular palettes. At each
vertex $v$, the number of palettes permitted to each party is chosen roughly
in proportion to that party's local degree. A shared random permutation
$\sigma_v$ of $[q]$ orders the palettes: Alice is permitted to use a prefix of
this ordering, while Bob is permitted to use a suffix. The two quotas are
chosen so that these two sets are disjoint.

An edge may be assigned to a regular palette only if that palette is permitted
to its owner at both endpoints. This rule immediately prevents conflicts
between the two parties. Indeed, if an Alice edge and a Bob edge assigned to
the same regular palette met at a vertex $v$, then that palette would have to
be permitted to both parties at $v$, contradicting the disjointness of their
permission sets. Thus, once the shared permissions are fixed, cross-party
conflicts within every regular palette are impossible, regardless of how the
individual edges are assigned.

The edge assignments can therefore use private randomness. If an Alice edge
$uv$ has at least one regular palette permitted at both endpoints, Alice
assigns it uniformly at random to one of these palettes; Bob independently
does the same for his own edges. These choices serve only to spread each
party's edges evenly among its permitted palettes, so they do not need to be
revealed to the other party. Consequently, the eventual LLL instance has only
one shared variable $\sigma_v$ per vertex, while all random choices associated
with individual edges remain private.

The only remaining issue is that an edge may have no regular palette permitted
at both endpoints. Each party places such edges in its own fallback graph, which each party
colors using its own separate palette. To keep the fallback graphs small, the
quota at every vertex includes a small additive cushion. This ensures that
even a low-degree vertex is permitted to use sufficiently many regular
palettes, so two endpoints are unlikely to have disjoint permission sets.  

There are therefore only two ways in which the construction can fail locally:
too many edges may be assigned to one regular palette at a vertex, or too many
edges may enter the fallback graph at a vertex. Standard concentration bounds
show that both events are exponentially unlikely in $\Delta$. We use these
local failures as the bad events in our LLL instance. If all of them are
avoided, every regular subgraph and both fallback graphs have sufficiently
small maximum degree to be colored within their allocated palettes using
Vizing's theorem.

The construction is well suited to our two-party LLL because it has only one
shared random variable per vertex, while all edge-level routing choices remain
private. Moreover, the bad events are local and exponentially unlikely. Our
LLL protocol therefore gives the desired communication and round bounds.

\subsection{Additional related work}
\label{sec:related}

Our work connects three active lines of research: the communication complexity of
graph problems, edge coloring with small palettes, and constructive and distributed
versions of the Lov\'asz local lemma. All three have a substantial history,
with considerable recent progress.

\paragraph{Graph problems in two-party communication.}
The study of graph problems in two-party communication goes back at least to
Hajnal, Maass, and Tur\'an~\cite{HajnalMaassTuran88}, who proved tight
deterministic communication bounds for connectivity, $s$--$t$ connectivity,
and bipartiteness. Ivanyos, Klauck, Lee, Santha, and de
Wolf~\cite{IvanyosKlauckLeeSanthaWolf12} subsequently studied classical and
quantum communication for a broader collection of graph properties,
including perfect matching, Eulerian tours, and triangle freeness.

This area has seen substantial renewed activity in recent years. Blikstad,
van den Brand, Efron, Mukhopadhyay, and
Nanongkai~\cite{BlikstadBrandEfronMukhopadhyayNanongkai22} obtained nearly
optimal bounds for maximum bipartite matching and related problems.
Blikstad, Jiang, Mukhopadhyay, and
Yingchareonthawornchai~\cite{BlikstadJiangMukhopadhyayYingchareonthawornchai25}
showed that global vertex connectivity has randomized communication
complexity $\widetilde\Theta(n^{3/2})$, while $s$--$t$ vertex connectivity
admits near-linear communication. Jiang, Nanongkai, and
Sawettamalya~\cite{JiangNanongkaiSawettamalya26} gave a deterministic
$\widetilde O(n^{11/7})$-bit protocol for minimum $s$--$t$ cut, and Cheng,
Jiang, Sawettamalya, and Yu~\cite{ChengJSY26} recently obtained simple
deterministic protocols for general matching, negative-cycle detection, and
shortest paths with negative edge weights. Round complexity has also received
attention: Radhakrishnan, Reddy, and Venkat~\cite{RadhakrishnanReddyVenkat26}
proved an $\Omega(n\log n)$ randomized lower bound for two-round
connectivity, together with superlinear lower bounds for every constant
number of rounds.

Graph coloring has only recently been studied in the two-party edge-partition
model. Flin and
Mittal~\cite{FlinMittal25} gave a zero-error randomized protocol for
$(\Delta+1)$-vertex coloring using $O(n)$ expected communication. They also
proved an $\Omega(n)$ lower bound for constant-error randomized protocols,
even for graphs of maximum degree two, thereby settling the randomized
communication complexity at $\Theta(n)$. Their protocol processes the
vertices sequentially and uses $O(n)$ rounds in expectation.

Chang, Mishra, Nguyen, and Salim~\cite{ChangMNS25} subsequently showed that
optimal communication can also be achieved with much fewer rounds. Their
randomized protocol still uses $O(n)$ expected communication, but reduces the
round complexity to $O(\log\log n\cdot\log\Delta)$ in the worst case. Thus,
$(\Delta+1)$-vertex coloring admits a protocol that is simultaneously
communication-optimal and round-efficient.

\paragraph{Edge coloring with small palettes.}
A central randomized approach to near-optimal edge coloring is the nibble
method, which gradually colors most of the edges while keeping the remaining
degrees under control. Together with applications of the LLL, this method has
led to efficient randomized $(1+\eps)\Delta$-edge-coloring algorithms in the
\textsf{LOCAL} model of distributed computing~\cite{DubhashiGrablePanconesi98,ElkinPettieSu15,ChangHLPU18}.

With a larger round complexity, distributed algorithms can also reach the
exact Vizing bound of $\Delta+1$ colors. Bernshteyn~\cite{Bernshteyn22} gave
the first deterministic distributed algorithm achieving this bound in $\poly(\Delta, \log n)$ rounds, and
subsequent work developed faster algorithms using multi-step Vizing chains
and entropy-compression techniques~\cite{Christiansen23,BernshteynDhawan25}.

There has also been rapid recent progress on centralized algorithms for
Vizing's theorem. A sequence of works substantially improved the running time
for computing a $(\Delta+1)$-edge coloring and culminated in randomized
near-linear-time and deterministic almost-linear-time
algorithms~\cite{Assadi25Vizing,BhattacharyaCarmonCostaSolomonZhang24,
BhattacharyaCostaSolomonZhang25,AssadiEtAl25NearLinear,
AssadiEtAl26Deterministic}. 

\paragraph{Constructive and distributed LLL.}
Moser and Tardos~\cite{MoserTardos10} introduced the resampling algorithm and
witness-tree analysis underlying our protocol. Haeupler, Saha, and
Srinivasan~\cite{HaeuplerSahaSrinivasan10} studied further properties of the
Moser--Tardos algorithm and used them to obtain new algorithmic applications
of the LLL.

A large body of subsequent work has studied the LLL in distributed models.
Chung, Pettie, and Su~\cite{ChungPettieSu17} gave an $O(\log n)$-round
randomized distributed algorithm, while Brandt~et~al.~\cite{BrandtEtAl16}
proved an $\Omega(\log\log n)$ randomized lower bound even for constant
dependency degree. A sequence of later works gave further improved
distributed LLL algorithms under various criteria
\cite{FischerGhaffari17,GhaffariHarrisKuhn18,BrandtMausUitto19,
BrandtGrunauRozhon20,Davies23,DaviesPeck25}.

\subsection{Organization}
\label{sec:organization}

In \Cref{sec:two-party-lll}, we develop and analyze the
communication-efficient constructive LLL protocol, proving
\Cref{thm:two-party-lll}. In \Cref{sec:edge-coloring}, we develop a randomized
edge-coloring procedure, apply our LLL protocol to it, and prove
\Cref{thm:main}. In \Cref{sec:applications}, we give several further
applications of our two-party LLL protocol. We conclude in
\Cref{sec:conclusion} with open problems.

\section{Constructive Lov\'asz local lemma}
\label{sec:two-party-lll}

In this section, we prove \Cref{thm:two-party-lll}. Recall that the partitioned
LLL model and the parameters $(N,M,p,D,h)$ were defined in
\Cref{sec:results-lll}. 
We begin by describing the protocol, which implements Moser--Tardos resampling
using public resampling tables. We then analyze why the protocol is efficient.
First, witness trees show that reaching a late round requires a long chain of
dependent resamplings, so the expected number of resamplings decreases
geometrically across rounds. Second, we encode only the net changes to the
shared resampling counters, converting this geometric decay into the
communication bound in \Cref{thm:two-party-lll}.

\subsection{Moser--Tardos resampling and public tables}
\label{sec:prelim-mt}

The Moser--Tardos algorithm~\cite{MoserTardos10} starts from an independent
sample of all variables. It then repeatedly selects any bad event $B$ that
currently occurs and resamples every variable in $\vbl(B)$ independently from
its original distribution. The choice among the currently occurring events
may follow any fixed or adaptive rule.

We implement this randomness using resampling tables. For every shared
variable $X_i$, where $i\in[N]$, public randomness generates an infinite
sequence
$$
X_i(0),X_i(1),X_i(2),\ldots
$$
of independent samples from the distribution of $X_i$. Both parties maintain
a counter $c_i$, initially zero, and interpret $X_i(c_i)$ as the current value
of $X_i$. Whenever $X_i$ is resampled, its counter increases by one. Each
private variable has an analogous table and counter, visible only to its
owner.

Since the table entries are independent samples consumed in order, this
implementation has exactly the same distribution as ordinary Moser--Tardos
resampling. Its advantage is that the parties can synchronize a shared
variable without transmitting its value. Once both parties know its counter,
they can recover the same value from the public table, regardless of the size
of the variable's domain.

\subsection{The communication protocol}
\label{sec:lll-protocol}

The protocol alternates between Alice and Bob. Round $1$ belongs to Alice,
round $2$ belongs to Bob, and so on. In a round, the active party repeatedly
resamples one of its currently occurring bad events until none of its bad
events occurs. We call such a round \emph{exhaustive}. The party may use any
fixed or adaptive rule to choose which occurring event to resample. 

Since local computation is unrestricted, one might ask why the
active party repeatedly resamples bad events instead of directly
finding an assignment that avoids all its bad events. The reason
is that, by using resampling, the two parties together carry out
a valid execution of the Moser--Tardos algorithm. This allows us
to use its analysis to bound the communication and the number
of alternating rounds.

At the end of a round, the active party communicates how much each shared
counter increased. Formally, for round $j$, let $c_i^{\rm before}$ and
$c_i^{\rm after}$ be the values of counter $c_i$ before and after the round,
and define
$$
u_i^{(j)}:=c_i^{\rm after}-c_i^{\rm before}.
$$
The vector $u^{(j)}=(u_1^{(j)},\ldots,u_N^{(j)})$ is the \emph{net
shared-counter update} of round $j$. The active party sends an encoding of
this vector, using the efficient encoding described in
\Cref{sec:generic-communication-proof1}. The other party then updates its copy
of the shared counters and reconstructs the current shared values from the
public tables before beginning the next round.

Only the net update is communicated. For example, during an Alice round, Bob
waits until Alice has finished before checking any of his bad events. He
therefore needs to know only the final shared values, not the intermediate
values produced during Alice's resamplings. The same reasoning applies during
a Bob round. Moreover, the complete sequence of
resamplings performed by the protocol is a valid Moser--Tardos execution:
every resampled bad event occurs when it is selected, and the Moser--Tardos
algorithm permits an arbitrary selection rule.

For round $j$, let $Z_j$ be the number of resamplings performed during the
round, and let
$$
U_j:=\left\lVert u^{(j)}\right\rVert_1
=\sum_{i=1}^N u_i^{(j)}
$$
be the total number of shared-counter increments. If round $j$ is not invoked
because the protocol has already terminated, we set $Z_j:=U_j:=0$. Since
each resampling involves at most $h$ shared variables, it increments at most
$h$ shared counters. Therefore,
\begin{equation}
\label{eq:generic-UZ}
U_j\le hZ_j.
\end{equation}

It remains to specify when the protocol stops. The protocol always performs
the first two rounds. Starting from round $2$, it stops at the end of the
first round $j$ with $U_j=0$. The following observation shows that the
resulting assignment avoids every bad event.

\begin{observation}[Zero-update termination]
\label{lem:generic-termination}
If $j\ge2$ and $U_j=0$, then no bad event occurs at the end of round $j$.
\end{observation}

\begin{proof}
Suppose that round $j$ is an Alice round; the other case is symmetric. Since
the round is exhaustive, none of Alice's bad events occurs at its end. Bob
finished round $j-1$ with none of his bad events occurring. During Alice's
round, Bob's private variables do not change, and $U_j=0$ means that no shared
variable is resampled. Thus, all variables on which Bob's bad events depend
retain their values from the end of round $j-1$, so none of Bob's bad events
occurs at the end of round $j$.
\end{proof}

Thus, whenever the protocol stops, its output avoids all bad events. We next
bound the expected number of resamplings across rounds and use this estimate
to derive the claimed round and communication bounds.

\subsection{Witness trees}
\label{sec:prelim-witness}

Resampling a bad event can make another bad event occur, which
may then require further resampling. To analyze this process,
fix a particular step at which a bad event $B$ is resampled.
We look backwards through the execution and collect earlier
resamplings that may have led to $B$ occurring at this step.
Moser and Tardos~\cite{MoserTardos10} introduced \emph{witness trees}
to organize such resampling histories: the root is labeled by
$B$, and the other nodes represent earlier resamplings.

We recall their witness-tree construction and probability bound
below. We will then show that a resampling in a late round of
our protocol must have a deep witness tree, allowing us to bound
the number of communication rounds.

We begin by defining the class of labeled trees used in the
construction. For a bad event $B$, let
$\Gamma^+(B):=\Gamma(B)\cup\{B\}$ denote its inclusive neighborhood.

\begin{definition}[Proper trees]
A rooted tree whose nodes are labeled by bad events is
\emph{proper} if, for every node labeled $B$, each of its children
has a label in $\Gamma^+(B)$, and no two of its children have
the same label. For a node $u$ in a proper tree, write $B(u)$
for the bad event labeling $u$.
\end{definition}

We now associate with each resampling a proper tree, called its
\emph{witness tree}. The \emph{resampling log} is the sequence of
bad events in the order in which they are resampled. Consider
any finite prefix $B_1,B_2,\ldots,B_t$ of this log. We construct
the witness tree associated with the resampling at time $t$
as follows.

Start with a root labeled $B_t$, and then scan
$B_{t-1},B_{t-2},\ldots,B_1$ in reverse order. When the scan reaches $B_i$,
ignore it if its label does not belong to the inclusive neighborhood of any
label currently in the tree. Otherwise, add a new node labeled $B_i$ and
attach it as a child of a deepest node whose label belongs to
$\Gamma^+(B_i)$. If there are several such nodes at the same depth, choose one arbitrarily. We say that the resulting tree \emph{occurs} as a witness tree in the
resampling log.

The resulting witness tree is proper. The condition on parent and child
labels follows directly from the attachment rule. To see that siblings have
distinct labels, suppose that a node labeled $B$ has already been added when
an earlier occurrence of $B$ is encountered. The existing node is eligible
as a parent, so the new node is attached below a node whose depth is at least
that of the existing node. The two nodes labeled $B$ therefore cannot be
siblings.

The depth of a node is the number of edges on its path from the root, and the
depth of a tree is the maximum depth of any of its nodes. Thus, the root has
depth zero. Attaching each new node below a deepest eligible node ensures that
a long chain of dependent resamplings cannot be represented by a shallow
tree. This property allows us to connect witness-tree depth to the number of
alternating rounds.

\begin{proposition}[Moser--Tardos witness-tree bound~\cite{MoserTardos10}]
\label{prop:mt-witness}
For any fixed proper tree $T$,
$$
\Prb[T\text{ occurs as a witness tree in the resampling log}]
\le\prod_{u\in V(T)}\Prb[B(u)].
$$
\end{proposition}

Distinct resampling times produce distinct witness trees. Trees with different
root labels are clearly distinct. If two resamplings have the same root label
$B$, the witness tree of the later resampling contains the earlier resampling
of $B$ and therefore has strictly more nodes labeled $B$. Consequently,
summing the bound in \Cref{prop:mt-witness} over any family of proper witness
trees bounds the expected number of resamplings whose witness trees belong to
that family. This argument applies to every adaptive event-selection rule,
including the rules used during the exhaustive rounds of our protocol.
\subsection{Witness-tree tail and round complexity}
\label{sec:lll-depth}

We now prove two estimates used to bound the number of communication rounds.
Recall that $\lambda:=ep(D+1)\le1-\delta$. The first follows by combining the
Moser--Tardos witness-tree bound~\cite{MoserTardos10} with a standard counting
argument. We include the proof for completeness. It shows that the expected
number of resamplings whose witness tree has large depth decreases
geometrically with the depth. The second estimate connects witness-tree depth
to our alternating protocol by showing that every resampling in a late round
must have a deep witness tree.

\begin{lemma}[Deep witness trees~\cite{MoserTardos10}]
\label{lem:generic-deep-witness}
Fix any event-selection rule and any integer $k\ge0$.  Over the randomness of
the resampling tables, the expected number of resamplings in the entire
execution whose witness tree has depth at least $k$ is at most
$$
\frac{Mp}{1-\lambda}\lambda^k.
$$
\end{lemma}

\begin{proof}
Put $d:=D+1$ and first fix the root label.  The inclusive neighborhood of any
event contains at most $d$ labels.  Fix an ordering of each inclusive
neighborhood $\Gamma^+(B)$, and associate each event in $\Gamma^+(B)$ with a
distinct child slot in $[d]$.  Consider a node of a witness tree labeled $B$.
Its children have distinct labels in $\Gamma^+(B)$, so we may place each child
in the slot associated with its label.  Applying this rule recursively embeds
the witness tree into the infinite rooted $d$-ary tree, in which every node
has $d$ distinguished child slots.

This embedding is injective once the root label is fixed.  Indeed, the slot of
a child and the label of its parent determine the child's label. 
The number of $t$-vertex rooted subtrees of the infinite $d$-ary tree is the
Fuss--Catalan number~\cite[Equation~(7.66)]{GrahamKnuthPatashnik94}
$$
T_d(t)=\frac{1}{(d-1)t+1}\binom{dt}{t}
=\frac{1}{t}\binom{dt}{t-1}.
$$
For $t\ge2$, we have
\begin{align*}
T_d(t)
&\le\frac{1}{t}\left(\frac{edt}{t-1}\right)^{t-1}
&&\text{since }\binom{n}{r}\le\left(\frac{en}{r}\right)^r \\
&\le(ed)^{t-1}.
\end{align*}
Indeed, $(t/(t-1))^{t-1}\le e\le t$ for $t\ge3$, while $t=2$ follows
directly; the case $t=1$ is immediate.  Hence, for any fixed root label, there
are at most $[e(D+1)]^{t-1}$ witness trees with $t$ nodes.

A tree of depth at least $k$ has at least $k+1$ nodes.  There are at most $M$ choices
for its root, and \Cref{prop:mt-witness} bounds the probability of each
$t$-node tree by $p^t$.  Therefore,
\begin{align*}
\EE[\text{number of resamplings with depth at least }k]
&\le M\sum_{t\ge k+1}[e(D+1)]^{t-1}p^t \\
&=Mp\sum_{t\ge k+1}\lambda^{t-1}
&&\text{since }\lambda=ep(D+1) \\
&=\frac{Mp}{1-\lambda}\lambda^k.
\end{align*}
Distinct resampling times have distinct witness trees, so every resampling is
associated with a different tree in the sum.  Hence the expected number
of resamplings is at most the sum of the occurrence probabilities of
the relevant witness trees, which is the quantity bounded above.
\end{proof}

The next lemma is the two-party analogue of the standard depth-versus-round
property for parallel Moser--Tardos resampling
\cite[Lemma~4.1]{MoserTardos10}.  The only additional issue is that the two
parties have private variables that the other party cannot change.

\begin{lemma}[A late round forces depth]
\label{lem:generic-round-depth}
Every resampling in round $j\ge3$ has witness-tree depth at least $j-2$.
\end{lemma}

\begin{proof}
Consider a resampling $r$ in an Alice round $j$; the Bob case is symmetric.

\paragraph{Within round $j$.}
Trace the causal history of $r$ backwards within the round.  If its event did
not occur at the beginning of the round, some earlier dependent resampling in
the same round must have caused it to occur.  Repeating this argument, with
the resampling time decreasing at every step, reaches an Alice event $B$ that
already occurred when round $j$ began.

\paragraph{Crossing round $j-1$.}
The event $B$ did not occur when Alice's preceding round $j-2$ ended, because
Alice processed her events exhaustively.  During the intervening Bob round,
no Alice-private variable changed.  Therefore, some Bob resampling in round
$j-1$ must have changed a shared variable on which $B$ depends.  That
resampling is dependent on $B$.

\paragraph{From the causal chain to tree depth.}
Applying the same argument recursively produces a time-decreasing chain that
contains a resampling from every round $j,j-1,\ldots,2$, with consecutive
events dependent.  When the witness tree of $r$ is constructed backwards, the
later event of each consecutive pair is already present when the earlier one
is scanned.  The earlier event is attached below a deepest eligible node, so
its depth is at least one larger.  The chain has at least $j-2$ links, and the
witness tree therefore has depth at least $j-2$.
\end{proof}

Recall that $A:=Mp/(1-\lambda)$. For $j\ge3$, every resampling in round $j$
has witness-tree depth at least $j-2$ by
\Cref{lem:generic-round-depth}, so \Cref{lem:generic-deep-witness} bounds the
expected number of such resamplings. For each of the first two rounds, we
apply the same witness-tree bound with $k=0$, which bounds the expected number
of resamplings in the entire execution by $A$. Therefore,
\begin{equation}
\label{eq:generic-round-decay}
\EE[Z_j]\le
\begin{cases}
A,&j=1,2,\\
A\lambda^{j-2},&j\ge3.
\end{cases}
\end{equation}

\begin{lemma}[Round complexity]
\label{lem:generic-rounds}
The expected number of rounds of the protocol is
$$
O\left(1+\frac{\log(1+A)}{\log(1/\lambda)}\right).
$$
\end{lemma}

\begin{proof}
If round $j\ge3$ is invoked, then round $j-1$ did not terminate and hence
$Z_{j-1}\ge1$. By \Cref{eq:generic-round-decay} and Markov's inequality, the
probability that round $j$ is invoked is at most
$\min\{1,A\lambda^{j-3}\}$. The expected number of rounds is therefore at
most
\begin{align*}
2+\sum_{k\ge0}\min\{1,A\lambda^k\}&=O\left(1+\frac{\log(1+A)}{\log(1/\lambda)}\right).
\end{align*}
If $A\ge1$, the summands equal $1$ for at most
$1+\lfloor\log A/\log(1/\lambda)\rfloor$ values of $k$.
If $A<1$, no summand equals $1$.
In either case, the remaining summands form a geometric
series with first term at most $1$, whose sum is at most
$1/(1-\lambda)\le1/\delta$.
This proves the claimed bound.
\end{proof}

\subsection{Encoding the net update vectors}
\label{sec:generic-communication-proof1}

The only information that the two parties need to communicate in round $j$ is the net shared-counter update vector
$u^{(j)}$. We now show how to encode this vector using a number of bits
bounded in terms of the total number $U_j$ of shared-counter increments in
the round. We first bound the expected length of one message in terms of
$\EE[U_j]$ and then sum these bounds using the geometric decay from
\Cref{eq:generic-round-decay}.

There are two simple ways to communicate the update, but neither is efficient
for all values of $U_j$. Sending the entire vector
$u^{(j)}\in\mathbb Z_{\ge0}^N$ coordinate-by-coordinate requires
$\Omega(N)$ bits, which is wasteful when only a few counters change. At the
other extreme, the sender could list the index of a counter once for every
increment. This uses $O(U_j\log N)$ bits and works well when $U_j$ is small,
but it becomes too expensive when $U_j$ is large. For example, when
$U_j=\Theta(N)$, the cost is $\Theta(N\log N)$, exceeding the linear
communication bound that we seek. We therefore use an encoding that remains
efficient throughout the full range of $U_j$.

The sender first transmits one control bit indicating whether $U_j=0$. If
$U_j=0$, no further data are sent. Otherwise, let $u:=U_j>0$. The sender
transmits $u$ using $O(\log(1+u))$ bits. The update vector is then a
nonnegative integer vector satisfying
$$
u_1^{(j)}+\cdots+u_N^{(j)}=u.
$$
By the stars-and-bars correspondence, such a vector can be represented by a
string containing $u$ stars and $N-1$ bars, where the numbers of stars before,
between, and after the bars give its coordinates. Hence there are
$$
\binom{N+u-1}{u}
$$
possible update vectors with coordinate sum $u$.

For each positive integer $u$, the parties fix a public ordering of these
vectors. Once the receiver knows $u$, the sender transmits the rank of
$u^{(j)}$ in this ordering using
$\left\lceil\log\binom{N+u-1}{u}\right\rceil$ bits. The receiver can then
recover every coordinate of $u^{(j)}$. Thus, a single rank specifies both
which counters changed and how much each counter increased.

\begin{lemma}[Cost of one round]
\label{lem:generic-one-round-code}
Let $\mu_j:=\EE[U_j]$. If $\mu_j=0$, no update data are sent. Otherwise,
excluding the one control bit sent in each invoked round, the expected length
of the message in round $j$ is
$$
O\left(
\mu_j\log\left(e\left(1+\frac{N}{\mu_j}\right)\right)
\right).
$$
\end{lemma}

\begin{proof}
Condition on $U_j=u>0$. The encoding uses $O(\log(1+u))$ bits for $u$ and at most
$\left\lceil\log\binom{N+u-1}{u}\right\rceil$ bits for the rank. The latter term
satisfies
\begin{align*}
\log\binom{N+u-1}{u}
&\le u\log\left(\frac{e(N+u-1)}{u}\right)
&&\text{since $\binom{n}{r}\le(en/r)^r$}\\
&\le u\log\left(e\left(1+\frac{N}{u}\right)\right).
\end{align*}
Moreover, $\log(1+u)\le u$ for every positive integer $u$, so the cost of
encoding $u$ is absorbed into the same bound. Thus, the message length is
$$
O\left(
u\log\left(e\left(1+\frac{N}{u}\right)\right)
\right).
$$

For $u>0$, define
$f(u):=u\log\left(e\left(1+N/u\right)\right)$, and set $f(0):=0$, which
agrees with its limit as $u$ approaches zero. A direct calculation gives
$
f''(u)=-\frac{N^2}{(\ln 2)\,u(u+N)^2}<0
$
for $u>0$, so $f$ is concave on $[0,\infty)$. Jensen's inequality now gives
$
\EE[f(U_j)]\le f(\EE[U_j])
=\mu_j\log\left(e\left(1+\frac{N}{\mu_j}\right)\right)
$, 
which proves the claimed bound.
\end{proof}

\subsection{Communication complexity}
\label{sec:generic-communication-proof}

We have bounded the expected number of update-data bits in round $j$ in terms
of $\mu_j=\EE[U_j]$. We now sum this bound over all rounds and then account
for the control bit sent in each invoked round.

\begin{lemma}[Geometric update coding]
\label{lem:geometric-coding}
For $N\ge1$, $x > 0$, and $0<\lambda\le1-\delta$,
$$
\sum_{k\ge0}
x\lambda^k\log\left(e\left(1+\frac{N}{x\lambda^k}\right)\right)
=O\left(x\log\left(2+\frac{N}{x}\right)\right).
$$
\end{lemma}

\begin{proof}
Let
$y_k:=x\lambda^k$.

\paragraph{The case $x\le N$.}
In this case, $N/x\ge1$, and hence
$$
1+\frac{N}{y_k}
=1+\frac{N}{x}\lambda^{-k}
\le2\frac{N}{x}\lambda^{-k}.
$$
Taking logarithms and multiplying by $y_k=x\lambda^k$ gives
$$
y_k\log\left(e\left(1+\frac{N}{y_k}\right)\right)
=O\left(
x\lambda^k\left(
1+\log\frac{N}{x}+k\log\frac{1}{\lambda}
\right)
\right).
$$
We use the standard formulas
$$
\sum_{k\ge0}\lambda^k=\frac{1}{1-\lambda}
\qquad\text{and}\qquad
\sum_{k\ge0}k\lambda^k=\frac{\lambda}{(1-\lambda)^2}.
$$
Summing the preceding bound over $k\ge0$ therefore gives
\begin{align*}
\sum_{k\ge0}y_k\log\left(e\left(1+\frac{N}{y_k}\right)\right)
&=O\left(
\frac{x}{1-\lambda}\left(1+\log\frac{N}{x}\right)
+x\frac{\lambda\log(1/\lambda)}{(1-\lambda)^2}
\right)\\
&=O\left(x\left(1+\log\frac{N}{x}\right)\right)\\
&=O\left(x\log\left(2+\frac{N}{x}\right)\right).
\end{align*}
The second line follows from $1-\lambda\ge\delta$ and
$\lambda\log(1/\lambda)/(1-\lambda)^2=O(1)$.

\paragraph{The case $x>N$.}
Let $k_0$ be the first index for which $y_{k_0}\le N$. For $k<k_0$, we have
$N/y_k<1$, so each summand is $O(y_k)$. These terms therefore contribute
$$
O\left(\sum_{k<k_0}x\lambda^k\right)=O(x).
$$
The terms beginning at $k_0$ form a geometric sequence whose first term is
$y_{k_0}\le N$. Applying the first case to this sequence shows that the tail
contributes
$$
O\left(
y_{k_0}\log\left(2+\frac{N}{y_{k_0}}\right)
\right)
=O(N).
$$
Since $x>N$, this is also $O(x)$. Finally,
$x\log(2+N/x)=\Theta(x)$ when $x>N$, which proves the claim.
\end{proof}

\begin{lemma}[Communication cost]
\label{lem:generic-communication}
The expected number of bits exchanged by the protocol is
$$
O\left(
1+Mph\log\left(2+\frac{N}{Mph}\right)
\right).
$$
\end{lemma}

\begin{proof}
If $N=0$, there are no shared counters, so no update data are sent and the
protocol terminates after its first two rounds. Hence, assume $N\ge1$.

Recall that $A=Mp/(1-\lambda)$ and let $\mu_j=\EE[U_j]$. By
\Cref{eq:generic-UZ,eq:generic-round-decay},
$$
\mu_j\le
\begin{cases}
hA,&j=1,2,\\
hA\lambda^{j-2},&j\ge3.
\end{cases}
$$
The function
$u\mapsto u\log\left(e\left(1+N/u\right)\right)$ is nondecreasing. Thus,
\Cref{lem:generic-one-round-code} bounds the first two messages by the
corresponding expression evaluated at $hA$, while the bounds for the
subsequent messages decrease geometrically. Applying
\Cref{lem:geometric-coding}, the expected total number of  bits communicated is
$$
O\left(hA\log\left(2+\frac{N}{hA}\right)\right)
=O\left(Mph\log\left(2+\frac{N}{Mph}\right)\right),
$$
excluding the control bits indicating whether $U_j=0$.
Here we used $hA=Mph/(1-\lambda)$ and
$\delta\le1-\lambda\le1$. 

Each invoked round also sends one control bit. Since $Mph>0$, we have
$h\ge1$, and therefore
$A=Mp/(1-\lambda)=O(Mph)$. By \Cref{lem:generic-rounds} and
$\lambda\le1-\delta$, the expected number of rounds is
$O(1+\log(1+A))=O(1+Mph)$. The expected number of control bits satisfies the
same bound. Finally, the logarithmic factor in the statement is at least one,
so the control bits are covered by the claimed communication bound.
\end{proof}

We can now conclude the proof of \Cref{thm:two-party-lll}.

\thmtwoparty*

\begin{proof}
The protocol is defined in \Cref{sec:lll-protocol}.  Taking $k=0$ in \Cref{lem:generic-deep-witness} shows that the expected total
number of resamplings is at most $A<\infty$. Hence only finitely many
resamplings occur almost surely, so
\Cref{lem:generic-termination} shows that the protocol terminates almost
surely and never outputs an invalid assignment.
The expected number of rounds is bounded by \Cref{lem:generic-rounds}, and the
expected number of bits exchanged is bounded by
\Cref{lem:generic-communication}.  This proves the theorem.
\end{proof}

\section{Edge coloring}
\label{sec:edge-coloring}

We now apply our two-party LLL protocol to edge coloring. Let $G=(V,E)$ be an
$n$-vertex simple graph of maximum degree at most $\Delta$, whose edge set is
partitioned as $E=E_A\cup E_B$. Alice knows $E_A$ and Bob knows $E_B$, and
each party must color its own edges so that together they form a proper edge
coloring of $G$. For a vertex $v$, let $\deg_A(v)$ and $\deg_B(v)$ denote its
degrees in the two edge sets.

The main challenge is to let Alice and Bob reuse essentially the same color
space without knowing each other's edges. Our idea is to divide most of the
colors into a small number of \emph{regular palettes}. At each vertex, shared
randomness determines which regular palettes Alice is allowed to use and
which Bob is allowed to use. These two sets are always disjoint.

To color an edge $uv$, a party may use only a regular palette that is
permitted at both $u$ and $v$. Thus an Alice edge and a Bob edge assigned to
the same regular palette can never share an endpoint: at any common endpoint,
that palette would have to be permitted to both parties, contradicting the
disjointness of their permitted sets there. This rules out cross-party
conflicts within every regular palette.

It may happen, however, that an edge $uv$ has no regular palette permitted at
both endpoints. Each party places such an edge in its \emph{fallback graph}.
Alice and Bob color their fallback graphs using separate palettes, so fallback
edges belonging to different parties cannot create a color conflict.

We now define the palettes formally. Let $P_1,\ldots,P_q$ be pairwise
disjoint regular palettes, and let $P_A^F$ and $P_B^F$ be two additional
palettes, disjoint from each other and from all regular palettes. Each
regular palette contains $L+1$ colors, while each fallback palette contains
$\lfloor\xi\Delta\rfloor+1$ colors. The parameters $q$, $L$, and $\xi$ are
specified in \Cref{sec:edge-parameters}.

\subsection{The construction}
\label{sec:construction}

We first describe the complete randomized construction.  The construction has
four steps.

\begin{description}[leftmargin=2.7cm,style=nextline]
\item[Quotas.]
At each vertex $v$, we first determine how many of the $q$ regular palettes
Alice and Bob may use. Alice's quota is
$$
s_A(v):=\left\lceil
q\left(
\underbrace{
\frac{1}{1+4\xi}\frac{\deg_A(v)}{\Delta}
}_{\text{degree-proportional term}}
+
\underbrace{
\frac{\xi}{1+4\xi}
}_{\text{additive cushion}}
\right)
\right\rceil,
$$
and Bob's quota is defined symmetrically by
$$
s_B(v):=\left\lceil
q\left(
\frac{1}{1+4\xi}\frac{\deg_B(v)}{\Delta}
+
\frac{\xi}{1+4\xi}
\right)
\right\rceil.
$$
The first term makes each party's quota roughly proportional to its local
degree. The additive cushion ensures that even a party with small
local degree is permitted to use sufficiently many regular palettes.

The parameters are chosen so that $s_A(v)+s_B(v)\le q$, allowing Alice's
and Bob's sets of permitted palettes to be disjoint. Each party can
compute its own quota from its local degree, without learning the other
party's degree.

\item[Routing.]
For every Alice edge $e$, independently sample a uniformly random permutation
$\pi_e^A$ of $[q]$.  For $e=uv$, let
$I_A(e):=S_A(u)\cap S_A(v)$ be the set of indices $j$ for which Alice is
permitted to use $P_j$ at both endpoints. If $I_A(e)\ne\varnothing$, Alice
assigns $e$ to $P_j$, where $j$ is the first element of $I_A(e)$ according to
$\pi_e^A$. Conditioned on $I_A(e)$, this index is uniformly distributed over
$I_A(e)$. If $I_A(e)=\varnothing$, Alice places $e$ in her fallback graph.
Bob applies the symmetric rule, using independent random permutations
$\pi_e^B$ for his edges.

For every $j\in[q]$, let $G_j^A$ be the subgraph consisting of Alice's edges
assigned to $P_j$, and let $G_A^F$ be the subgraph consisting of Alice's
fallback edges. Define $G_j^B$ and $G_B^F$ symmetrically.

\item[Coloring.]
After the LLL has fixed the random choices, Alice colors $G_j^A$ using $P_j$
for every $j\in[q]$, while Bob colors $G_j^B$ using the same palette. Alice
colors $G_A^F$ using $P_A^F$, and Bob colors $G_B^F$ using $P_B^F$. All these
colorings are computed locally using Vizing's theorem.
\end{description}

The important point is that Alice and Bob can safely reuse each regular
palette. Suppose that an edge of $G_j^A$ and an edge of $G_j^B$ shared an
endpoint $v$. Then $j\in S_A(v)\cap S_B(v)$, contradicting the disjointness
of the two permission sets. Thus, no vertex is incident to edges in both
$G_j^A$ and $G_j^B$.

It therefore suffices to choose the random variables so that, for each party,
every graph $G_j^A$ and $G_j^B$ has maximum degree at most $L$, while $G_A^F$
and $G_B^F$ have maximum degree at most $\xi\Delta$. Under these conditions,
Vizing's theorem guarantees that the corresponding palettes are large enough.
These are exactly the conditions that we enforce using the LLL.

\subsection{Parameters}
\label{sec:edge-parameters}

We choose the parameters so that all palettes fit within the available budget
of $(1+\eps)\Delta$ colors. The $q$ regular palettes are shared by Alice and
Bob, and each contains $L+1$ colors. In addition, each party has a separate
fallback palette containing $\lfloor\xi\Delta\rfloor+1$ colors. Thus, the
total number of colors is
$$
q(L+1)+2\left(\lfloor\xi\Delta\rfloor+1\right).
$$
Our parameter choices will ensure that this quantity is at most
$(1+\eps)\Delta$, while also ensuring that the permission quotas fit and
fallback edges are sufficiently rare.

Set
$$
\xi:=\frac{\eps}{20},
\qquad
q:=
\left\lceil
\left(\frac{1+4\xi}{\xi}\right)^2
\ln\frac{4e}{\xi}
\right\rceil,
\qquad
L:=\left\lceil\frac{(1+4\xi)\Delta}{q}\right\rceil.
$$
Here $\xi$ determines the size of the fallback palettes, $q$ is the number of
regular palettes, and $L$ is the maximum degree that each regular palette is
designed to handle.

\begin{lemma}[The quotas fit]
\label{lem:quota-fit}
For every vertex $v$, $s_A(v)+s_B(v)\le q$, so 
$S_A(v)\cap S_B(v)=\varnothing$.
\end{lemma}

\begin{proof}
Recall that
\begin{align*}
s_A(v)
&:=\left\lceil
q\left(
\frac{1}{1+4\xi}\frac{\deg_A(v)}{\Delta}
+\frac{\xi}{1+4\xi}
\right)
\right\rceil, \\
s_B(v)
&:=\left\lceil
q\left(
\frac{1}{1+4\xi}\frac{\deg_B(v)}{\Delta}
+\frac{\xi}{1+4\xi}
\right)
\right\rceil.
\end{align*}
Using $\lceil x\rceil\le x+1$ for each quota, we obtain
\begin{align*}
s_A(v)+s_B(v)
&\le
\frac{q}{1+4\xi}
\left(\frac{\deg_A(v)+\deg_B(v)}{\Delta}+2\xi\right)+2
&&\text{removing the two ceilings} \\
&\le \frac{q(1+2\xi)}{1+4\xi}+2
&&\text{since }\deg_A(v)+\deg_B(v)\le\Delta \\
&= q-\frac{2\xi q}{1+4\xi}+2
  \\
&\le q
&&\text{since }\frac{2\xi q}{1+4\xi}\ge2.
\end{align*}
The last inequality follows from the definition of $q$. Therefore,
$s_A(v)+s_B(v)\le q$.

Finally, $S_A(v)$ is a prefix of $\sigma_v$ of length $s_A(v)$, while
$S_B(v)$ is a suffix of length $s_B(v)$. Since their total length is at most
$q$, these two sets are disjoint.
\end{proof}

We next verify the color budget. Recall that each regular palette $P_j$ has
$L+1$ colors, while each fallback palette has
$\lfloor\xi\Delta\rfloor+1$ colors.

\begin{lemma}[The color budget fits]
\label{lem:palette}
The palettes $P_1,\ldots,P_q,P_A^F,P_B^F$ can be chosen pairwise disjoint
using at most $(1+\eps)\Delta$ colors, provided that $\Delta$ is
at least some sufficiently large constant depending on $\eps$.
\end{lemma}

\begin{proof}
Since $L\le(1+4\xi)\Delta/q+1$, the regular palettes use at most
$$
q(L+1)\le (1+4\xi)\Delta+2q
$$
colors. The two fallback palettes use at most
$$
2\left(\lfloor\xi\Delta\rfloor+1\right)\le2\xi\Delta+2
$$
colors. Hence, the total number of colors is at most
$$
(1+6\xi)\Delta+2q+2.
$$
Here the term $2q$ accounts for both the rounding in $L$ and the extra color
required by Vizing's theorem for the subgraphs using each regular palette.
Since $\xi=\eps/20$, we have $6\xi=3\eps/10$.  Moreover, $q=q(\eps)$ is a
constant.  Thus, when $\Delta$ is at least some sufficiently large constant
depending only on $\eps$, we have $2q+2\le 7\eps\Delta/10$.  The total is then
at most $(1+\eps)\Delta$.
\end{proof}

Finally, we record the bounds used to control fallback edges. The additive cushion in each quota gives, for every vertex $v$,
\begin{equation}
\label{eq:quota-lower-bound}
s_A(v)\ge\frac{\xi q}{1+4\xi}
\qquad\text{and}\qquad
s_B(v)\ge\frac{\xi q}{1+4\xi}.
\end{equation}
Thus, each party receives a constant fraction of the regular palettes at
every vertex. Moreover, the choice of $q$ gives
\begin{equation}
\label{eq:q-fallback}
\exp\left(-\left(\frac{\xi}{1+4\xi}\right)^2q\right)
\le\frac{\xi}{4e}.
\end{equation}
The two bounds above will be used in the probability analysis in \Cref{sec:analysis}.

\subsection{Probability analysis}
\label{sec:analysis}

Two types of overload can occur at a vertex. A regular palette is overloaded
if more than $L$ incident edges are routed to it, while the fallback palette
is overloaded if more than $\xi\Delta$ incident edges use it. For $v\in V$
and $j\in[q]$, define
$$
D_j^A(v):=\deg_{G_j^A}(v)
\qquad\text{and}\qquad
F_A(v):=\deg_{G_A^F}(v).
$$
Define $D_j^B(v)$ and $F_B(v)$ symmetrically. We show that each overload
probability is exponentially small in $\Delta$.

\begin{lemma}[Degree of the regular subgraphs]
\label{lem:overload-tail}
There exists a constant $c_1=c_1(\eps)>0$ such that, for every vertex $v$ and
every $j\in[q]$, both $\Prb[D_j^A(v)>L]$ and $\Prb[D_j^B(v)>L]$ are at most
$e^{-c_1\Delta}$.
\end{lemma}

\begin{proof}
We prove the statement for Alice.  If $j\notin S_A(v)$, then
$D_j^A(v)=0$.  Suppose therefore that $j\in S_A(v)$, and condition on
$\sigma_v$.

Consider an Alice edge $uv$ incident to $v$. By symmetry, every palette in
$S_A(v)$ has the same probability of being selected for $uv$: the permission
set $S_A(u)$ and the routing permutation
$\pi_{uv}^A$ are both generated uniformly and do not favor any palette in
$S_A(v)$. Since $uv$ is assigned to at most one regular palette, each of the
$s_A(v)$ palettes is selected with probability at most $1/s_A(v)$. Therefore,
for every $j\in S_A(v)$,
$$
\Prb[uv\text{ is assigned to }j\mid\sigma_v]
\le\frac{1}{s_A(v)}.
$$
For the fixed index $j$, write
$$
D_j^A(v)
=
\sum_{u:uv\in E_A}
\1[uv\text{ is assigned to }j].
$$
Conditioned on $\sigma_v$, the summands are mutually independent. Indeed,
whether $uv$ is assigned to $j$ is determined by $\sigma_u$ and
$\pi_{uv}^A$. Because $G$ is simple, different Alice edges incident to $v$
have different opposite endpoints $u$, and hence their assignments depend on
disjoint sets of independent random variables.

Let $\mu:=\EE[D_j^A(v)\mid\sigma_v]$. Summing the preceding probability
bound over the $\deg_A(v)$ Alice edges incident to $v$ gives
\begin{align*}
\mu
&\le\frac{\deg_A(v)}{s_A(v)} \\
&\le
\frac{(1+4\xi)\Delta}{q}
\cdot\frac{\deg_A(v)}{\deg_A(v)+\xi\Delta}
&&\text{since }
s_A(v)=\left\lceil
\frac{q\left(\deg_A(v)+\xi\Delta\right)}{(1+4\xi)\Delta}
\right\rceil \\
&\le\frac{(1+4\xi)\Delta}{q(1+\xi)}
&&\text{since }\deg_A(v)\le\Delta \\
&\le\frac{L}{1+\xi}
&&\text{since }L\ge\frac{(1+4\xi)\Delta}{q}.
\end{align*}
This calculation exhibits the first role of the additive cushion: it creates a
constant multiplicative gap between the expected load and the capacity $L$.

We use a standard multiplicative Chernoff bound: if $Y$ is a sum of
independent Bernoulli random variables with mean at most $\bar\mu$, then, for
every $\eta>0$,
$$
\Prb[Y\ge(1+\eta)\bar\mu]
\le\exp\left(-\frac{\eta^2\bar\mu}{2+\eta}\right).
$$
In our setting, $\bar\mu:=L/(1+\xi)$ is an upper bound on the conditional
mean $\mu$. Since $L=(1+\xi)\bar\mu$, applying the bound with $\eta=\xi$
gives
\begin{align*}
\Prb[D_j^A(v)>L\mid\sigma_v]
&\le\Prb[D_j^A(v)\ge L\mid\sigma_v]\\
&\le\exp\left(-\frac{\xi^2\bar\mu}{2+\xi}\right)\\
&=\exp\left(-\frac{\xi^2L}{(1+\xi)(2+\xi)}\right).
\end{align*}
Since $L\ge(1+4\xi)\Delta/q$ and both $\xi$ and $q$ depend only on $\eps$,
the last expression is at most $e^{-c_1\Delta}$ for some constant
$c_1=c_1(\eps)>0$. Averaging over $\sigma_v$ proves the unconditional bound.
The proof for Bob is identical.
\end{proof}

\begin{lemma}[Fallback degree]
\label{lem:fallback-tail}
There exists a constant $c_2=c_2(\eps)>0$ such that, for every vertex $v$,
both $\Prb[F_A(v)>\xi\Delta]$ and $\Prb[F_B(v)>\xi\Delta]$ are at most
$e^{-c_2\Delta}$.
\end{lemma}

\begin{proof}
Again we prove the statement for Alice and condition on $\sigma_v$.  Consider
an Alice edge $uv$, and write $s_u:=s_A(u)$ and $s_v:=s_A(v)$.  If
$s_u+s_v>q$, then $S_A(u)$ and $S_A(v)$ must intersect, so $uv$ cannot be a
fallback edge.  Suppose that $s_u+s_v\le q$.  Conditioned on $\sigma_v$, the
set $S_A(u)$ is a uniformly random $s_u$-subset of $[q]$.  Therefore,
\begin{align*}
\Prb[I_A(uv)=\varnothing\mid\sigma_v]
&=\frac{\binom{q-s_v}{s_u}}{\binom{q}{s_u}} \\
&\le\left(1-\frac{s_v}{q}\right)^{s_u}
&&\text{by sampling without replacement} \\
&\le\exp\left(-\frac{s_us_v}{q}\right)
&&\text{since }1-x\le e^{-x} \\
&\le\exp\left(-\left(\frac{\xi}{1+4\xi}\right)^2q\right)
   &&\text{by \eqref{eq:quota-lower-bound}} \\
&\le\frac{\xi}{4e}
&&\text{by \eqref{eq:q-fallback}}.
\end{align*}
This calculation shows the second role of the additive cushion: even when one
or both endpoints have small degree in Alice's subgraph, it keeps both
permission sets large enough that their intersection is unlikely to be empty.

For every Alice edge $uv$ incident to $v$, let
$$
Y_{uv}:=\1[I_A(uv)=\varnothing]
$$
indicate whether $uv$ is placed in the fallback graph. Then
$$
F_A(v)=\sum_{u:uv\in E_A}Y_{uv}.
$$
Under the conditioning on $\sigma_v$, the variable $Y_{uv}$ depends only on
$\sigma_u$. Because $G$ is simple, different Alice edges incident to $v$
have different opposite endpoints $u$. Their indicator variables therefore
depend on distinct independent permutations and are mutually independent.

The preceding calculation shows that
$\EE[Y_{uv}\mid\sigma_v]\le\xi/(4e)$ for every Alice edge $uv$ incident to
$v$. Hence, writing $\mu:=\EE[F_A(v)\mid\sigma_v]$, we have
$$
\mu\le\frac{\xi}{4e}\deg_A(v)\le\frac{\xi\Delta}{4e}.
$$
We use the following standard Chernoff bound: if $Y$ is a sum of independent
Bernoulli random variables, then, for every $x\ge\EE[Y]$,
$$
\Prb[Y\ge x]\le\left(\frac{e\EE[Y]}{x}\right)^x.
$$
If $\mu=0$, then $F_A(v)=0$ with probability one and the claim is immediate.
Otherwise, set $x:=\xi\Delta$. Since $\mu\le x/(4e)$, the Chernoff bound gives
\begin{align*}
\Prb[F_A(v)>\xi\Delta\mid\sigma_v]
&\le\Prb[F_A(v)\ge x\mid\sigma_v]\\
&\le\left(\frac{e\mu}{x}\right)^x
&&\text{by the Chernoff bound}\\
&\le4^{-x}
&&\text{since }\mu\le\frac{x}{4e}\\
&=e^{-(\xi\ln4)\Delta}.
\end{align*}
Averaging over $\sigma_v$ proves the claim with $c_2=\xi\ln4$. The proof for
Bob is identical.
\end{proof}

\subsection{The LLL instance}
\label{sec:lll}

We now define an LLL instance whose bad events correspond to violations
of the required degree bounds. For each vertex $v$, Alice's bad event
occurs if $v$ has degree greater than $L$ in one of her regular subgraphs
or greater than $\xi\Delta$ in her fallback graph. Formally, define
$$
\mathcal B_A(v):=
\left\{\max_{j\in[q]}D_j^A(v)>L\right\}
\cup
\left\{F_A(v)>\xi\Delta\right\},
$$
and define Bob's bad event $\mathcal B_B(v)$ symmetrically.
Thus, if no bad event occurs, every regular subgraph $G_j^A$ and $G_j^B$
has maximum degree at most $L$, and both fallback graphs $G_A^F$ and
$G_B^F$ have maximum degree at most $\xi\Delta$.

The shared variables are the vertex permutations $\sigma_v$. Alice's private
variables are the edge permutations $\pi_e^A$, and Bob's private variables
are the edge permutations $\pi_e^B$. Let $N_A(v)$ and $N_B(v)$ denote the
Alice and Bob neighborhoods of $v$, and write
$N_A[v]:=N_A(v)\cup\{v\}$ and $N_B[v]:=N_B(v)\cup\{v\}$.

To apply \Cref{thm:two-party-lll}, we specify five parameters: the number $N$
of shared variables, an upper bound $M$ on the number of bad events, an
upper bound $p$ on the probability of each bad event, an upper bound $h$
on the number of shared variables involved in any one bad event, and an
upper bound $D$ on the dependency degree.

\begin{lemma}[Parameters of the LLL instance]
\label{lem:edge-lll-parameters}
There exists a constant $c=c(\eps)>0$ such that, when $\Delta$ is at least
some sufficiently large constant depending only on $\eps$, the LLL instance
above has the valid parameter bounds
$$
N=n,\qquad M=2n,\qquad p=e^{-c\Delta},\qquad
h=\Delta+1,\qquad D=2\Delta^2+1.
$$
\end{lemma}

\begin{proof}
There is one shared permutation $\sigma_v$ for each vertex, so $N=n$. There
is one Alice bad event and one Bob bad event at each vertex, so $M=2n$.

We next bound the probability of a bad event. By
\Cref{lem:overload-tail,lem:fallback-tail} and a union bound,
\begin{align*}
\Prb[\mathcal B_A(v)]
&\le\sum_{j=1}^q\Prb[D_j^A(v)>L]
   +\Prb[F_A(v)>\xi\Delta] \\
&\le q e^{-c_1\Delta}+e^{-c_2\Delta}
&&\text{by \Cref{lem:overload-tail,lem:fallback-tail}} \\
&\le e^{-c\Delta}
&&\text{for some }c=c(\eps)>0
\end{align*}
when $\Delta$ is at least some sufficiently large constant depending only on
$\eps$.  The same bound holds for $\mathcal B_B(v)$.

We next count the shared variables involved in each event. The event
$\mathcal B_A(v)$ depends only on the shared permutations $\sigma_u$ for
$u\in N_A[v]$ and the private permutations $\pi_e^A$ for Alice edges $e$
incident to $v$.  It therefore involves at most
$\deg_A(v)+1\le\Delta+1$ shared variables.  The analogous statement holds for
Bob, so $h=\Delta+1$ is a valid upper bound.

Finally, consider the dependency degree of a bad event centered at
$v$. Any bad event dependent on it must be centered within distance two of
$v$. Indeed, if the two events share a vertex permutation $\sigma_x$, then
both centers are either $x$ or neighbors of $x$. If they share an edge
permutation, then their centers are the two endpoints of that edge.

There are at most
$1+\Delta+\Delta(\Delta-1)=\Delta^2+1$ vertices within distance two of $v$.
Each such vertex is the center of one Alice event and one Bob event. Excluding
the original event, there are therefore at most
$2(\Delta^2+1)-1=2\Delta^2+1$ dependent events. Hence
$D=2\Delta^2+1$ is a valid upper bound.
\end{proof}

\subsection{Applying the two-party LLL protocol}
\label{sec:edge-lll-invocation}

We now apply the two-party LLL protocol of \Cref{thm:two-party-lll} to the LLL instance above. We call the
execution of this protocol, during which the vertex and edge permutations may
be resampled, the \emph{resampling stage}. Once this stage terminates, no bad
event occurs, and the parties can color the resulting subgraphs locally.

It remains to check the LLL condition and substitute the parameters from
\Cref{lem:edge-lll-parameters} into the communication and round bounds of
\Cref{thm:two-party-lll}. These parameters are
$$
N=n,\qquad M=2n,\qquad p=e^{-c\Delta},\qquad
D=2\Delta^2+1,\qquad h=\Delta+1.
$$
The bad-event probability decreases exponentially in $\Delta$, while the
dependency degree grows only polynomially, so we have
\begin{equation}
\lambda :=ep(D+1)
=e\cdot e^{-c\Delta}(2\Delta^2+2) =e^{-\Omega(\Delta)}. \label{myeqq}
\end{equation}
Hence, when $\Delta$ is at least some sufficiently large constant depending
only on $\eps$, we have $\lambda\le1/2$. Thus,
\Cref{thm:two-party-lll} applies with constant slack.

\begin{lemma}[Cost of the resampling stage]
\label{lem:edge-cost}
When $\Delta$ is at least some sufficiently large constant depending only on
$\eps$, the resampling stage terminates almost surely, communicates
$O\left(ne^{-\gamma\Delta}+1\right)$ bits in expectation for some constant
$\gamma=\gamma(\eps)>0$, and uses
$O\left(\frac{\log n}{\Delta}+1\right)$ rounds in expectation.
\end{lemma}

\begin{proof}
The LLL condition was verified above, so \Cref{thm:two-party-lll} applies. We
evaluate its communication and round bounds separately.

The theorem bounds the expected communication by
$$
O\left(
1+Mph\log\left(2+\frac{N}{Mph}\right)
\right).
$$
For our parameters,
$$
Mph=2n(\Delta+1)e^{-c\Delta},
$$
and
$$
\log\left(2+\frac{N}{Mph}\right)
=\log\left(2+\frac{e^{c\Delta}}{2(\Delta+1)}\right)
=O(\Delta).
$$
Substituting these estimates gives
\begin{align*}
O\left(
1+Mph\log\left(2+\frac{N}{Mph}\right)
\right)
&=O\left(1+2n(\Delta+1)e^{-c\Delta}\cdot\Delta\right) \\
&=O\left(1+n\Delta^2e^{-c\Delta}\right) \\
&=O\left(ne^{-\gamma\Delta}+1\right)
&&\text{for some }\gamma=\gamma(\eps)>0.
\end{align*}
Here the last step absorbs the factor $\Delta^2$ into the exponential function by
changing the constant in the exponent.

For the round complexity, \Cref{thm:two-party-lll} gives
$$
O\left(1+\frac{\log(1+A)}{\log(1/\lambda)}\right),
\qquad
A:=\frac{Mp}{1-\lambda}.
$$
Since $\lambda\le1/2$,
$$
A\le2Mp=4ne^{-c\Delta},
$$
and hence $\log(1+A)=O(\log n)$. For the denominator,
\begin{align*}
\log\frac{1}{\lambda}
&=\Omega(\Delta) & \text{by \eqref{myeqq}}.
\end{align*}
Substituting these two estimates shows that the expected number of rounds is
$O\left(\frac{\log n}{\Delta}+1\right)$. Almost-sure termination follows
directly from \Cref{thm:two-party-lll}.
\end{proof}

\subsection{Completing the proof}

The construction above uses public randomness. To obtain the private-coin
guarantee in \Cref{thm:main}, we use the following Las Vegas version of
Newman's theorem; see~\cite{Newman91} and~\cite[Exercise~3.15]{KN97}.

\begin{lemma}[Newman's theorem]
\label{lem:newman-las-vegas}
Suppose a two-party computation problem has at most $K$ possible input
pairs and admits a public-coin Las Vegas protocol that communicates $C$
bits in expectation and uses an expected number of $R$ rounds. Then it
admits a private-coin Las Vegas protocol that communicates
$O(C+\log\log K)$ bits in expectation and uses an expected
number of $O(R+1)$ rounds.
\end{lemma}

We are now ready to prove \Cref{thm:main}.

\thmmain*

\begin{proof}
Choose $\Delta_0$ sufficiently large that the parameter estimates and the LLL
condition established above hold whenever $\Delta\ge\Delta_0$.

\paragraph{Public-coin cost.}
Use the random variables and routing rule from \Cref{sec:construction}. Before
coloring the routed subgraphs, perform the resampling stage defined in
\Cref{sec:edge-lll-invocation}. By \Cref{lem:edge-cost}, this stage terminates
almost surely using
$O(ne^{-\gamma\Delta} + 1)$ expected communication and
$O\left(\frac{\log n}{\Delta}+1\right)$ expected rounds, for some constant
$\gamma=\gamma(\eps)>0$.

\paragraph{Correctness.}
At termination, no bad event occurs. Hence, every graph
$G_j^A$ and $G_j^B$ has maximum degree at most $L$, while $G_A^F$ and $G_B^F$
have maximum degree at most $\xi\Delta$. Since these degrees are integers, the
two fallback graphs have maximum degree at most $\lfloor\xi\Delta\rfloor$.
By Vizing's theorem, $G_j^A$ and $G_j^B$ can each be colored using $P_j$, and
$G_A^F$ and $G_B^F$ can be colored using $P_A^F$ and $P_B^F$, respectively.

We check that these local colorings fit together. Two edges in the same routed
subgraph receive different colors whenever they share an endpoint because that
subgraph is properly colored. Two edges routed to different palettes receive
colors from disjoint sets. The only remaining possibility is an edge of
$G_j^A$ and an edge of $G_j^B$ that use the same regular palette. Such edges
cannot share an endpoint $v$, because that would imply
$j\in S_A(v)\cap S_B(v)$, whereas these two permission sets are disjoint.
Therefore, the union of the local colorings is a proper edge coloring of $G$.

By \Cref{lem:palette}, the total number of colors is at most
$(1+\eps)\Delta$. The coloring step is performed locally after
the resampling stage and requires no further communication. Every terminating
execution therefore outputs a valid coloring, so the protocol is Las Vegas.
This proves the public-coin part of the theorem.

\paragraph{Private randomness.}
An input pair is determined by assigning each
of the $\binom n2$ possible edges to Alice, to Bob, or to neither party.
Hence the number $K$ of possible input pairs satisfies
$K\le 3^{\binom n2}$, and therefore $\log\log K=O(\log n)$.
Applying \Cref{lem:newman-las-vegas} to the public-coin protocol above gives
a private-coin Las Vegas protocol with expected communication complexity
$$
O\left(ne^{-\gamma\Delta}+\log n\right).
$$
Its expected number of rounds remains
$O\left(\frac{\log n}{\Delta}+1\right)$, completing the proof.
\end{proof}
\section{Further applications}
\label{sec:applications}

The purpose of this section is to illustrate the broad applicability of our
two-party LLL protocol. We do not claim novelty for the applications
themselves. Indeed, for each problem considered in this section, we just take a standard LLL
formulation and examine the communication guarantee obtained from
\Cref{thm:two-party-lll} when the constraints are divided between Alice and
Bob.

In each application, the variables and their domains are public, while Alice
and Bob hold disjoint private sets of constraints. Their task is to agree on
one assignment satisfying all constraints, without first exchanging their
complete constraint descriptions. A linear communication bound does not
follow immediately from the number of variables, because the number of
constraints can be much larger. The key calculation in each case is that the
small probability of a bad event compensates for the number of bad events and
the number of shared variables involved in each one, giving $Mph=O(N)$.

We apply \Cref{thm:two-party-lll} directly to each problem, first obtaining a
public-coin protocol. Since the private inputs in the applications below come
from finite input spaces, we then apply \Cref{lem:newman-las-vegas} to obtain
private-coin protocols. In applying Newman's theorem, we fix all information
that is public to both parties and let $K$ count only the possible pairs of
private inputs.

\subsection{Satisfiability}

Let $x_1,\ldots,x_N$ be a public set of Boolean variables. Alice and Bob hold
disjoint private collections $\mathcal C_A$ and $\mathcal C_B$ of clauses,
each containing exactly $k$ distinct variables. Together, these clauses form
a $k$-CNF formula. Suppose that every variable occurs in at most $d$ clauses
in $\mathcal C_A\cup\mathcal C_B$, where $N$, $k$, and $d$ are known to both
parties. Their task is to agree on a satisfying assignment without exchanging
their complete clause sets.

When $d$ is large, the formula may contain far more than $N$ clauses. Let
$m$ denote the number of clauses. Since each clause contains $k$ variables,
while each variable appears in at most $d$ clauses, double-counting
variable--clause incidences gives $m\le Nd/k$. Thus, even communicating a
constant number of bits for each clause may require $\omega(N)$ bits of communication, and it is not immediate how the parties can find a satisfying
assignment without exchanging a large amount of information about their
clauses. The LLL formulation is what overcomes this difficulty. A uniformly
random assignment violates any fixed clause with probability only $2^{-k}$,
and this exponentially small probability compensates for the large value of
$d$ in the communication bound.

\begin{corollary}[Clause-partitioned $k$-SAT with bounded variable occurrences]
\label{cor:k-sat}
If
$
d\le\frac{2^{k-2}}{k}
$, 
then Alice and Bob can find a satisfying assignment using a private-coin Las
Vegas protocol with $O(N)$ expected bits of communication and $O(\log N)$
expected rounds.
\end{corollary}

\begin{proof}
If $d=0$, there are no clauses, and the claim is immediate. Hence, assume that
$d\ge1$. Assign every variable independently and uniformly at random,
treating all $N$ variables as shared. For each clause $C$, its owner defines
the bad event $B_C$ that $C$ is violated. Each bad event has probability
$p=2^{-k}$ and depends on $h=k$ shared variables.

Two bad events can be dependent only if their clauses share a variable. A
clause contains $k$ variables, and each of them occurs in at most $d-1$ other
clauses. We may therefore take $D=k(d-1)$. Since
$$
D+1=k(d-1)+1\le kd,
$$
the assumption on $d$ gives
$$
ep(D+1)
\le e2^{-k}kd
\le\frac{e}{4}<1.
$$
Thus, the LLL condition holds with constant slack.

Let $m:=|\mathcal C_A|+|\mathcal C_B|$ be the total number of clauses.
Double-counting variable--clause incidences gives $km\le Nd$, so we may use
the public upper bound $M:=Nd/k$. For these parameters,
$$
Mph
=\frac{Nd}{k}\cdot2^{-k}\cdot k
=Nd2^{-k}
\le\frac{N}{4k}
=O(N).
$$
This is the point at which the factor $2^{-k}$ compensates for the potentially
large number of clauses. By \Cref{thm:two-party-lll}, there is a public-coin
protocol using $O(N)$ expected bits of communication. Moreover,
$Mp\le Mph=O(N)$ and $ep(D+1)\le e/4$, so the expected number of rounds is
$O(\log N)$.

There are at most $\binom Nk2^k$ possible clauses, and each such clause can
belong to Alice, belong to Bob, or be absent. Hence the number $K$ of possible
input pairs satisfies
$$
K\le3^{\binom Nk2^k}.
$$
Since $k\le N$, we have $\log\log K=O(N)$. Applying
\Cref{lem:newman-las-vegas} gives the claimed private-coin protocol without
changing the asymptotic communication or round complexity.
\end{proof}

Gebauer, Szab\'o, and Tardos~\cite{GebauerST16} determined the asymptotic
threshold for $k$-SAT with bounded variable occurrences. The largest number
of occurrences per variable for which every $k$-CNF formula is satisfiable is
$(2/e+o(1))2^k/k$. Their guarantee uses the lopsided LLL with a carefully
chosen biased product distribution. Our simpler symmetric-LLL argument gives
the same asymptotic dependence on $k$, although it does not recover the
optimal constant.

\subsection{Hypergraph 2-coloring}

Let $V$ be a public vertex set, and let $k\ge1$ and $D\ge0$ be integers known
to both parties. Alice and Bob hold disjoint private families
$\mathcal E_A$ and $\mathcal E_B$ of $k$-element subsets of $V$. Together
they form the $k$-uniform hypergraph
$H=(V,\mathcal E_A\cup\mathcal E_B)$. The parties must agree on a red--blue
coloring of $V$ with no monochromatic hyperedge. Suppose every hyperedge
intersects at most $D$ other hyperedges.

Hypergraph 2-coloring is one of the classical applications of the LLL and
has been studied extensively from an algorithmic perspective. Radhakrishnan
and Srinivasan~\cite{RadhakrishnanSrinivasan00} obtained improved extremal
bounds in terms of the number of hyperedges, together with fast randomized
algorithms and deterministic parallel versions. Local computation algorithms
for the problem were developed in a sequence of
works~\cite{RubinfeldTamirVardiXie11,AlonRubinfeldVardiXie12,
MansourRubinsteinVardiXie12,DorobiszKozik23}.

When $D$ is large, the hypergraph may contain far more than $|V|$
hyperedges. Indeed, each hyperedge contains $k$ vertices, while each vertex
belongs to at most $D+1$ hyperedges. Double-counting vertex--hyperedge
incidences therefore gives
$|\mathcal E_A|+|\mathcal E_B|\le |V|(D+1)/k$. Thus, a protocol using
$O(|V|)$ bits cannot afford to communicate even a constant amount of
information about every hyperedge. The LLL formulation overcomes this
difficulty: under a uniformly random 2-coloring, a fixed hyperedge is
monochromatic with probability only $2^{1-k}$, and this exponentially small
probability compensates for the potentially large value of $D$.

\begin{corollary}[Partitioned hypergraph 2-coloring]
\label{cor:hypergraph}
If $D+1\le2^{k-3}$, then Alice and Bob can find a proper 2-coloring using a
private-coin Las Vegas protocol with $O(|V|)$ expected bits of communication
and $O(\log |V|)$ expected rounds.
\end{corollary}

\begin{proof}
Give every vertex $v$ one shared unbiased random bit
$X_v\in\{0,1\}$, where the two values represent the two colors. For each
hyperedge $e$, its owner defines the bad event $B_e$ that all vertices in $e$
receive the same color. Thus, $p=2^{1-k}$, every bad event depends on at most
$h=k$ shared variables, and the dependency degree is at most $D$. The LLL
condition holds with constant slack, since
\begin{align*}
ep(D+1)
&=e2^{1-k}(D+1)
&&\text{since $\Prb[B_e]=2^{1-k}$}\\
&\le\frac{e}{4}<1
&&\text{since $D+1\le2^{k-3}$.}
\end{align*}

Let $N:=|V|$ and $m:=|\mathcal E_A|+|\mathcal E_B|$. If a vertex belongs to
$r$ hyperedges, then each of these hyperedges intersects the other $r-1$.
Therefore, every vertex belongs to at most $D+1$ hyperedges. Double-counting
vertex--hyperedge incidences gives
\begin{align*}
km
&=\sum_{v\in V}\deg_H(v)\\
&\le N(D+1)
&&\text{since $\deg_H(v)\le D+1$ for every $v$.}
\end{align*}
Thus, the number of bad events is at most $N(D+1)/k$, so in
\Cref{thm:two-party-lll} we may use the public upper bound
$M:=N(D+1)/k$. Consequently,
\begin{align*}
Mph
&=N(D+1)2^{1-k}\\
&\le\frac{N}{4}
&&\text{since $D+1\le2^{k-3}$.}
\end{align*}
Thus, the exponentially small probability that a hyperedge is monochromatic
compensates for the potentially large number of hyperedges. Applying
\Cref{thm:two-party-lll} gives a public-coin protocol using $O(N)$ expected
bits. Moreover, $ep(D+1)\le e/4<1$ and $Mp\le Mph=O(N)$, so the expected
number of rounds is $O(\log N)$.

There are $\binom Nk$ possible hyperedges, and each is owned by Alice, owned
by Bob, or absent. Thus, the total number of possible input pairs is $K\le3^{\binom Nk}$, so $\log\log K=O(N)$.
Applying \Cref{lem:newman-las-vegas} preserves the $O(N)$ communication bound
and gives the claimed private-coin protocol.
\end{proof}

More generally, the calculation above gives
$Mph\le N(D+1)2^{1-k}$. Hence, when $D+1=o(2^k)$, the more precise
communication bound in \Cref{thm:two-party-lll} is $o(N)$. However, converting the
protocol to a private-coin one adds $O\left(\log\binom Nk\right)$ bits, which
may dominate this sublinear bound.
\subsection{List coloring with bounded color degree}

Let $V$ be a public set of $n$ vertices. Alice is given an edge set $E_A$,
and Bob is given a disjoint edge set $E_B$. We assume that their union forms
a simple graph
$
G=(V,E_A\cup E_B)
$.
Each vertex $v\in V$ has a nonempty finite public list $P(v)$ of admissible
colors. For $v\in V$ and $c\in P(v)$, define the \emph{color degree}
$$
d_c(v):=\left|\left\{u\in N_G(v):c\in P(u)\right\}\right|.
$$
Thus, $d_c(v)$ counts the neighbors of $v$ that could conflict with $v$ if
$v$ receives color $c$. This quantity can be much smaller than the degree of
$v$ when the lists of neighboring vertices have little overlap.

Color-degree conditions for list coloring are classical. Haxell
\cite{Haxell01} proved that lists of size $2d$ suffice when every color degree
is at most $d$. Reed and Sudakov~\cite{ReedSudakov02} later showed that lists
of size $(1+o(1))d$ suffice as $d\to\infty$. These results give stronger
existence guarantees than the elementary symmetric LLL used below. Our
purpose here is to examine the communication bound obtained from the standard
LLL formulation.

\begin{corollary}[Bounded-color-degree public-list coloring]
\label{cor:color-degree}
Let $d\ge1$ and $Q\ge8d$ be integers known to both parties. Suppose
$d_c(v)\le d$ for every vertex $v$ and every color $c\in P(v)$, and
$|P(v)|\ge Q$ for every $v$. Then Alice and Bob can find a proper list
coloring using a private-coin Las Vegas protocol with
$
O\left(
1+\frac{nd}{Q}\log\left(2+\frac{Q}{d}\right)
\right)
=O(n)
$ 
expected bits of communication and $O(\log n)$ expected rounds.
\end{corollary}

\begin{proof}
If $Q\ge n$, Alice and Bob can use the public lists to greedily assign a
distinct color to every vertex without communication. Hence, we may assume
that $Q<n$ for the rest of the proof.

By truncating the public lists, we may assume that $|P(v)|=Q$ for every
vertex $v$. This cannot increase any color degree. For each vertex $v$, let
$X_v$ be a shared random variable chosen uniformly from $P(v)$. For every
edge $uv$ and every color $c\in P(u)\cap P(v)$, define the bad event
$$
B_{uv,c}:=\{X_u=X_v=c\}.
$$
The owner of $uv$ owns all events $B_{uv,c}$. Avoiding these events is
equivalent to finding a proper list coloring. Each event has probability
$p=1/Q^2$ and depends on the two shared variables $X_u$ and $X_v$, so we may
take $h=2$.

We next bound the dependency degree $D$. The number of bad events involving a
fixed variable $X_v$ is
\begin{align*}
\sum_{u\in N_G(v)}|P(u)\cap P(v)|
&=\sum_{c\in P(v)}d_c(v)\\
&\le Qd
&&\text{since $|P(v)|=Q$ and $d_c(v)\le d$.}
\end{align*}
Every bad event involves two variables and is therefore dependent on at most
$2Qd-2$ other bad events. We may thus take
$
D:=2Qd-2
$.
The LLL condition holds with constant slack, since
\begin{align*}
ep(D+1)
&\le\frac{e(2Qd-1)}{Q^2}\\
&<\frac{2ed}{Q}
\le\frac{e}{4}<1
&&\text{since $Q\ge8d$.}
\end{align*}

Let $m$ be the number of bad events. Each event $B_{uv,c}$ involves one
variable at each endpoint of $uv$, so double-counting event--variable
incidences gives
\begin{align*}
2m
&=\sum_{v\in V}\sum_{c\in P(v)}d_c(v)\\
&\le nQd
&&\text{since each of the $nQ$ color degrees is at most $d$.}
\end{align*}
Thus, in \Cref{thm:two-party-lll}, we may use the public upper bound
$M:=nQd/2$. Since the number of shared variables is $N=n$, while
$p=1/Q^2$ and $h=2$, we obtain
\begin{align*}
Mph
&=\frac{nQd}{2}\cdot\frac{1}{Q^2}\cdot2\\
&=\frac{nd}{Q}.
\end{align*}
Substituting these parameters into \Cref{thm:two-party-lll} gives a
public-coin protocol using
$$
O\left(
1+\frac{nd}{Q}\log\left(2+\frac{Q}{d}\right)
\right)
$$
expected bits of communication. Moreover, $ep(D+1)\le e/4$ and
$Mp=nd/(2Q)=O(n)$, so the expected number of rounds is $O(\log n)$.

It remains to remove the public randomness. Fix the public vertex set and
lists. An input pair is determined by assigning each of the
$\binom n2$ possible edges to Alice, to Bob, or to neither party. Hence, the
number $K$ of possible input pairs satisfies
$
K\le3^{\binom n2}
$,
and therefore $\log\log K=O(\log n)$. Applying
\Cref{lem:newman-las-vegas} adds $O(\log n)$ expected bits and preserves the
asymptotic round bound.

We show that the additional $O(\log n)$ expected bits are covered by the
claimed bound. Since $d\ge1$ and $Q<n$, and the function
$x\mapsto\log(2+x)/x$ is decreasing for $x>0$, we have
$$
\frac{nd}{Q}\log\left(2+\frac{Q}{d}\right)
\ge\log(2+n)
=\Omega(\log n).
$$
Thus, the private-coin conversion does not change the asymptotic
communication bound.
\end{proof}

We highlight that the communication cost becomes sublinear when the lists are large compared with
the color-degree bound. Indeed, if $Q/d=\omega(1)$, then
$$
\frac{nd}{Q}\log\left(2+\frac{Q}{d}\right)=o(n).
$$

As a special case, suppose every vertex has the same list $[Q]$. The color
degree is then at most $\Delta$, so taking $Q=8\Delta$ gives an alternative
proof that $O(\Delta)$-vertex coloring requires only $O(n)$ communication.
Previous work~\cite{FlinMittal25,ChangMNS25} obtains the same communication
bound using the smaller palette of $\Delta+1$ colors.

\section{Conclusion and open problems}
\label{sec:conclusion}
We have shown that a $(1+\eps)\Delta$-edge coloring can be
computed with surprisingly little coordination between the two parties. With public randomness, the
expected communication is \emph{sublinear} when $\Delta=\omega(1)$ and \emph{constant} when $\Delta\ge C_\varepsilon\log n$,
for a sufficiently large constant $C_\varepsilon$. Converting the protocol to use only private randomness costs an
additional $O(\log n)$ expected bits.

A main ingredient of this result is our communication-efficient constructive LLL for two-party partitioned instances. The theorem applies under the symmetric LLL condition with constant slack and extends beyond edge coloring, as illustrated by the additional applications in this paper. We expect that the partitioned LLL will be useful for other problems in which locally defined constraints are distributed between two parties.

We conclude with several open problems.

\paragraph{Approaching Vizing's bound.}
How far can the palette size be pushed toward Vizing's bound $\Delta+1$ without substantially increasing the communication?
Can one obtain a $\Delta+O(1)$ or $\Delta+1$ edge coloring with $O(n)$ expected communication?
More ambitiously, is sublinear communication possible in this regime when $\Delta$ grows with $n$?
It would also be interesting to remove the assumption that $\Delta$ is at least a sufficiently large constant.

\paragraph{Communication as a function of the degree.}
For fixed $\eps>0$, our upper bound shows that the communication complexity of $(1+\eps)\Delta$-edge coloring decreases rapidly as $\Delta$ grows, while the $\Omega(n)$ lower bound of~\cite{ChangMNS25} applies to constant-degree graphs.
A natural goal is therefore to determine the tight communication complexity as a function of both $n$ and $\Delta$.
In particular, is the dependence $ne^{-\Omega(\Delta)}$ in our upper bound optimal?
The same question is interesting for $\Delta+1$ vertex coloring.
Its randomized communication complexity is known to be $\Theta(n)$ in general~\cite{FlinMittal25}, but the known linear lower bound is again witnessed by bounded-degree graphs.
It remains open whether the communication complexity decreases with $\Delta$ for vertex coloring as well.

\paragraph{Deterministic protocols.}
Can the sublinear communication achieved by our edge-coloring protocol when $\Delta=\omega(1)$ also be obtained deterministically?
Both our constructive LLL and the edge-coloring application use randomness in an essential way.
A deterministic analogue does not seem to follow directly from our approach and may require substantially different ideas.

\section*{AI disclosure}
We used ChatGPT (OpenAI) and Claude (Anthropic) extensively in the development and writing of this paper. The core proof ideas came from ChatGPT and were subsequently developed into complete proofs and written up for the final manuscript through an iterative collaboration among the human authors, ChatGPT, and Claude. The human authors have independently verified the correctness and originality of all results in the paper and take full responsibility for its contents.

\bibliographystyle{alpha}
\bibliography{refs}

\end{document}